\pdfoutput=1

\documentclass[11pt]{article}

\usepackage[]{naacl2021}

\usepackage{times}
\usepackage{latexsym}

\usepackage{microtype}
\usepackage{times}
\usepackage{url}
\usepackage{latexsym}
\usepackage{helvet} 
\usepackage{courier}  
\usepackage{graphicx} 
\usepackage{graphicx}  
\usepackage{amssymb,amsmath,amsthm}
\usepackage{booktabs}
\usepackage{algorithm}
\usepackage{algorithmic}
\usepackage{bm} 
\usepackage{atbegshi,picture}

\AtBeginShipoutNext{\AtBeginShipoutUpperLeft{%
  \put(\dimexpr\paperwidth-1cm\relax,-1.5cm){\color{red} \makebox[0pt][r]{\framebox{The official version is published in NAACL-HLT 2021.}}}%
}}

\newcommand{\squishlist}{
   \begin{list}{$\bullet$}
    { \setlength{\itemsep}{0pt}      \setlength{\parsep}{0pt}
      \setlength{\topsep}{0pt}       \setlength{\partopsep}{0pt}
      \setlength{\leftmargin}{1em} \setlength{\labelwidth}{1em}
      \setlength{\labelsep}{0.5em} } }
\newcommand{\squishlisttwo}{
   \begin{list}{$\bullet$}
    { \setlength{\itemsep}{0pt}    \setlength{\parsep}{0pt}
      \setlength{\topsep}{0pt}     \setlength{\partopsep}{0pt}
      \setlength{\leftmargin}{2em} \setlength{\labelwidth}{1.5em}
      \setlength{\labelsep}{0.5em} } }
\newcommand{\squishlistend}{
    \end{list}  }
\newtheorem{theorem}{Theorem}

\newcommand*{\QEDA}{\hfill\ensuremath{\square}}

\DeclareMathOperator*{\concat}{concat}
\DeclareMathOperator*{\argmax}{argmax}
\DeclareMathOperator*{\cosine}{cosine}

\theoremstyle{definition}
\newtheorem{definition}{Definition}

\def\dd{{\mathbb{D}}}
\def\xg{{\tilde{x}}}
\def\xgi{{\tilde{x}_i}}
\def\gru{{\text{GRU}}}

\def\vp{v^{\prime}}

\def\xgdpi{{\varepsilon_{\epsilon, \rho_i}(\tilde{x}_i)}}
\def\xgdp{{\tilde{x}_{dp}}}

\usepackage[T1]{fontenc}

\usepackage[utf8]{inputenc}

\usepackage{microtype}

\title{ER-AE: Differentially Private Text Generation for Authorship Anonymization}

\author{Haohan Bo \\
  McGill University, Canada \\
  \texttt{haohan.bo@mail.mcgill.ca} \\\And
  Steven H. H. Ding \\
  Queen's University, Canada \\
  \texttt{ding@cs.queensu.ca} \\\AND
  Benjamin C. M. Fung \\
  McGill University, Canada \\
  \texttt{ben.fung@mcgill.ca} \\\And
  Farkhund Iqbal \\
  Zayed University, UAE \\
  \texttt{farkhund.iqbal@zu.ac.ae} \\
}
  
\date{}

\begin{document}
\maketitle
\begin{abstract}
Most of privacy protection studies for textual data focus on removing explicit sensitive identifiers. However, personal writing style, as a strong indicator of the authorship, is often neglected. Recent studies, such as SynTF, have shown promising results on privacy-preserving text mining. However, their anonymization algorithm can only output numeric term vectors which are difficult for the recipients to interpret. We propose a novel text generation model with a two-set exponential mechanism for authorship anonymization. By augmenting the semantic information through a REINFORCE training reward function, the model can generate differentially private text that has a close semantic and similar grammatical structure to the original text while removing personal traits of the writing style. It does not assume any conditioned labels or paralleled text data for training. We evaluate the performance of the proposed model on the real-life peer reviews dataset and the Yelp review dataset. The result suggests that our model outperforms the state-of-the-art on semantic preservation, authorship obfuscation, and stylometric transformation. 
\end{abstract}

\section{Introduction}
Privacy has become a vital issue in online data gathering and public data release. 
Various machine learning models and privacy preservation algorithms have been studied for relational data~\cite{johnson2018towards},
network graph data~\cite{chen2014correlated}, 
and transactional data \cite{li2012privbasis}. 
Some of them have been successfully adopted in real-life applications such as telemetry collection~\cite{cortes2016differential}. 
However, the studies on privacy protection for textual data are still preliminary. 
Most related works only focus on replacing the sensitive key phrases in the text
~\cite{vasudevan2014review}
without considering the author's writing style, which is indeed a strong indicator of a person's identity. 
Even though some textual data, such as double-blind academic reviews, is released anonymously, the adversaries may recover the author's identity using the personal traits in writing. 
Stylometric techniques~\cite{koppel2011authorship} can identify an author of the text from 10,000 candidates. 
They are effective across online posts, articles, emails, and reviews~\cite{ding2015visualizable,ding2017learning}.
Nevertheless, traditional text sanitization methods~\cite{narayanan2008robust} focus on anonymizing the contents, such as patient information, instead of the writing style, so they are ineffective against writing style analysis. 
The original author can be easily re-identified even if protected by these traditional approaches~\cite{IHFD08di,iqbal2010mining,iqbal2013unified,SIF15diin}.

Only a few recent studies focus on authorship anonymization, aiming to hide the personal traits of writing style in the given textual data. \emph{Anonymouth}~\cite{mcdonald2012use} is a semi-automatic framework that offers suggestions to users to change their writing style. Yet, this framework is not practical since it requires two datasets as a reference to compare the change in writing style. Also, the user has to make all the final modification decisions. \emph{SynTF}~\cite{weggenmann2018syntf} represents a line of research that protects the privacy of the numeric vector representation of textual data. It adopts the exponential mechanism for a privacy guarantee, but the output is only an opaque term frequency vector, not an interpretable text in natural language. Furthermore, its token substitution approach does not consider the grammatical correctness and semantic.

Style transfer is another line of research that tries to generate text with controllable attributes
~\cite{shen2017style,hu2017toward,sennrich2016controlling}.
Representative models
~\cite{hu2017toward}
can control the sentiment and tense of the generated text. However, they do not modify the personal traits in writing. Their applications on sentiment and word-reordering correspond to the content of the text more than the writing style. We argue that their definition of styles, such as sentiment or tense, is different from the personal linguistic writing characteristics that raise privacy concern. 
\emph{A4NT}~\cite{shetty20184} is a generative neural network that sanitizes the writing style of the input text. However, it requires text samples to be labeled with known author identities. It is not applicable to many textual data publishing scenarios.
Additionally, according to the samples provided in the paper, it has difficulties keeping the same semantic meaning between the original and the generated text. 
Without using any privacy model, A4NT does not provide any privacy guarantee. 

To address the aforementioned issues, we propose an \emph{Embedding Reward Auto-Encoder (ER-AE)} to generate differentially private text. 
Relying on differential privacy, it protects the author's identity through text indistinguishability without assuming any specific labels, any parallel data or any assumption on the attacker. It guards the privacy of the data against the worst information disclosure scenario. 
ER-AE receives the original text as input and generates a new text using the two-set exponential mechanism. 
We propose a \emph{REINFORCE}~\cite{sutton2000policy} embedding reward function to augment the semantic information during the text generation process. The model can keep the generated text a close semantic and sentiment similarity to the original while providing a guarantee that one can hardly recover the original author's identity. 
Unlike the aforementioned authorship anonymization works, ER-AE produces human-friendly text in natural language. Our key contributions are summarized as follows:
\squishlist
    \item 
    The first differentially private authorship anonymization model that can generate human-friendly text in natural language, instead of a numeric vector.
    \item 
    A novel two-set exponential mechanism to overcome the large output space issue while producing meaningful results.
    
    \item 
    A novel combination of a differential privacy mechanism with a sequential text generator, providing a privacy guarantee through a sampling process. 

    \item
    A new REINFORCE reward function that can augment the semantic information through external knowledge, enabling better preservation of the semantic similarity in the data synthesis process. 

    \item Comprehensive evaluations on two real-life datasets, namely \textit{NeurIPS \& ICLR peer reviews} and \textit{Yelp product reviews}, show that ER-AE is effective in obfuscating the writing style, anonymizing the authorship, and preserving the semantics of the original text. 
\squishlistend
All the source code and data are publicly accessible for reproducibility and transferability.\footnote{\href{https://github.com/McGill-DMaS/Authorship-Anonymization}{https://github.com/McGill-DMaS/Authorship-Anonymization}}


\section{Related Work}
\label{sec:relate}

\textbf{Differential Privacy.}
Recently, differential privacy has received a lot of attention in the machine learning community.
The deep private auto-encoder~\cite{phan2016differential} is designed to preserve the training data privacy. Their purpose is to guarantee that publishing the trained model does not reveal the privacy of individual records. Our purpose is different. We publish the differentially private data generated by the model, rather than the model itself. Most existing models for differentially private data release, such as Chen et al.~\shortcite{chen2014correlated}
, focus on different types of data rather than text. One recent work~\cite{weggenmann2018syntf} aims to protect privacy in text data using the exponential mechanism. However, it releases the term frequency vectors instead of a readable text. This approach limits the utility of published data to only the applications that assume term frequency as features. In contrast, our goal is to generate differentially private text in a natural language without compromising individual privacy.

\textbf{Writing Style Transfer.}
Studies on writing style transferal try to change the writing style revealed from the text according to a given author. Shetty et al.~\shortcite{shetty20184} design a GAN to transfer Obama's text to Trump's style. 
A sequence to sequence (seq2seq) model is proposed by Jhamtani et al.~\shortcite{jhamtani2017shakespearizing} to transfer modern English into Shakespearean English. 
Shetty et al.\shortcite{shen2017style}~design a model with a cross-alignment method to control the text sentiment while preserving semantic. 
These models can also be applied to writing style anonymization.
However, these studies require the data to be labeled with authorship identity. They assume a number of known authors. In contrast, ours does not assume any label information.


\textbf{Writing Style Obfuscation.} Writing style obfuscation studies try to hide the identity of an author. 
Anonymouth~\cite{mcdonald2012use} is a tool that utilizes \emph{JStylo} to generate writing attributes. It gives users suggestions on which way they can anonymize their text according to two reference datasets. 
\cite{kacmarcik2006obfuscating} also propose a similar architecture to anonymize text. However, instead of directly changing the text, they all work on the term frequency vector, whose real-life utility is limited. Compared with semi-automatic methods that require users to make a decision, our approach provides an end-to-end solution that directly learns from data.

\section{Preliminaries and Problem Definition}

Adjacency is a key notion in differential privacy. One of the commonly used adjacency definitions is that two datasets $D_{1}$ and $D_{2}$ are adjacent if $D_{2}$ can be obtained by modifying one record in $D_{1}$ \cite{dwork2010boosting}. Differential privacy~\cite{dwork2006calibrating} is a framework that provides a rigorous privacy guarantee on a dataset.  It demands \emph{inherent randomness} of a sanitization algorithm or generation function:


\begin{definition} \textbf{Differential Privacy.} Two datasets are considered as adjacent if there is only one single element is different. Let privacy buget $\epsilon > 0$, a randomized algorithm $\mathcal{A}: D^{n} \xrightarrow{} Z $, and the image of $\mathcal{A}$: $im(\mathcal{A})$. The algorithm $\mathcal{A}$ is said to preserve $\epsilon$-differential privacy if for any two  adjacent datasets $D_{1}$, $D_{2}$ $\in D^{n}$, and for any possible set of output $Z \in im(\mathcal{A})$:
\label{def:def}
\begin{equation*}
\label{dp}
Pr\left[\mathcal{A}\left(D_{1}\right) \in Z\right] \leq e^{\epsilon} \cdot Pr\left[\mathcal{A}\left(D_{2}\right) \in Z\right]\tag*{\QEDA}
\end{equation*}
\end{definition} 
It guarantees that the result from a given algorithm $\mathcal{A}$ is not sensitive to
a change of any individual record in $D$. 
$\epsilon$ denotes the privacy budget, the allowed degree of sensitivity. 
A large $\epsilon$ implies a higher risk to privacy. However, $\epsilon$ is a relative value that implies different degrees of risk given different problems~\cite{weggenmann2018syntf}. 
Some studies~\cite{sala2011sharing} use a large $\epsilon$, while the others~\cite{chen2014correlated} use a smaller value. 

\textbf{Adversary Scenario.} Generally in an authorship identification problem, one assumes that the attacker holds an anonymous text authored by one of the suspects from the dataset. The attacker aims to infer the true author of the anonymous text based on a set of reference texts from each suspect. However, this scenario assumes certain information on the applicable dataset, such as author labels and the number of reference text samples. Therefore, following~\cite{weggenmann2018syntf}, we define that any two pieces of text as adjacent datasets.

\textbf{Adjacency.} Any two pieces of text can be considered \textit{adjacent} in the strictest scenario that datasets $D_{1}$ and $D_{2}$ both have only one record, and $D_{2}$ can be obtained by editing one record in $D_{1}$ following~Definition~\ref{def:def}. 
With differential privacy, we can have text indistinguishability: one cannot distinguish the identity of any text to another. In our case, the identity of a text corresponds to the author who wrote the text. Along with this, the attacker would fail in the original authorship identification scenario since the anonymous text is indistinguishable from the rest of the dataset.

Our definition follows Weggenmann  and  Kerschbaum~\shortcite{weggenmann2018syntf}'s idea, leading to the strictest and most conservative definition of adjacency.
\begin{definition} \textbf{Differentially Private Text Generation.}
\label{def:dptg}
Let $\dd$ denote a dataset that contains a set of texts where $x \in \dd$ is one of them. $|x|$, the length of the text, is bound by $l$. Given $\dd$ with a privacy budget $\epsilon$, for each $x$ the model generates another text $\xgdp$ that satisfies $\epsilon l$-differential privacy. $\QEDA$

Following the above definitions, any two datasets that contain only one record are probabilistically indistinguishable w.r.t. a privacy budget $\epsilon$. It directly protects the identity of an individual record, disregarding whether some of the records belong to the same author or not. It assumes that every record is authored by a different author, which is the strictest situation.
Technically, the proposed text generation approach protects the writing style by reorganizing the text, replacing tokens with different spelling, removing the lexical, syntactical and idiosyncratic features of the given text. The above definition is based on SynTF~\shortcite{weggenmann2018syntf}, but our target is readable text rather than numeric vectors, which is more challenging.
 \end{definition}

\begin{figure*}[t]

\centering

\includegraphics[width=1 \textwidth, height=0.31 \textwidth]{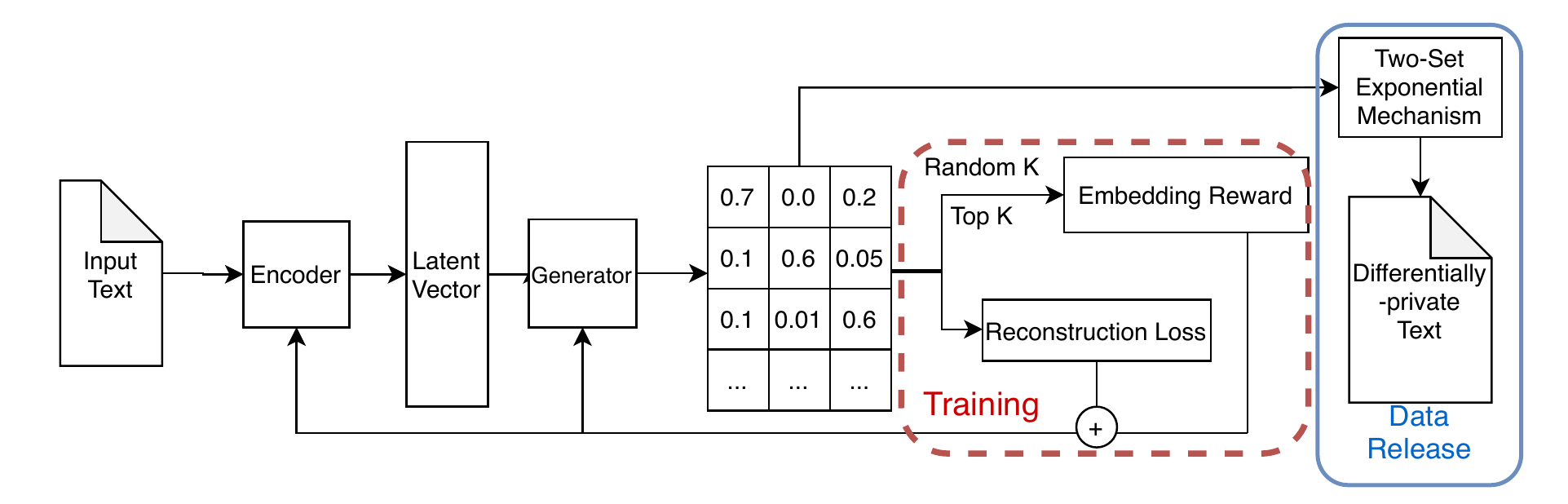}
\caption{Overall architecture of ER-AE.}
\label{fig:model_fig}
\end{figure*}

\section{ER-AE for Differentially Private Text Generation}
\label{Method}
Figure~\ref{fig:model_fig} depicts the overall architecture of our proposed ER-AE model, which consists of an encoder and a generator. The encoder receives a sequence of tokens as input and generates a latent vector to represent the semantic features. The generator, which is incorporated with the two-set exponential mechanism, can produce differentially private text according to the latent vector. ER-AE is trained by combining a reconstruction loss function and a novel embedding loss function.

\begin{algorithm}
\caption{Generation Procedure of ER-AE}
\label{alg:er-ae}
\small
\begin{algorithmic} 
 \STATE \textbf{Input}: Text: $x$, Parameters: $\theta$, Encoder: $E_{\theta}()$, Generator: $G_{\theta}()$, Privacy budget: $\epsilon$.
 \STATE  Produce the latent vector: $E_{\theta}(x)$.
 \STATE    Get probabilities of new tokens: $Pr[\xg] \leftarrow G_{\theta}(E_{\theta}(x))$.
  \FOR{$i \leftarrow 1$ to length of $x$}
 \STATE  Build two candidate token sets based on $Pr[\xgi]$: $S$, $O$.
 \STATE   Apply exponential mechanism to choose token set: $T$.
 \STATE  Randomly sample new $i$-th token from $T$: $\xgdp[i]$.
\ENDFOR
  \STATE \textbf{Output}: Differentially Private Text: $\xgdp$.
\end{algorithmic}
\end{algorithm}

Our ER-AE model starts with a basic sequence-to-sequence (seq2seq) auto-encoder structure. 
Given a text $x$, its tokens $\langle  x_{1} \dots x_{l}  \rangle$ are firstly converted into a sequence of embedding vectors  $\langle  Em(x_{1}) \dots Em(x_{l}) \rangle$ by $Em: \mathcal{V} \rightarrow \mathbb{R}^{m_{1}}$, where $\mathcal{V}$ is the vocabulary across the dataset and $m_1$ is the embedding dimension. On its top, we apply a bi-directional recurrent neural network with \emph{Gated Recurrent Unit (GRU)}~\cite{cho2014properties} that leverages both the forward and backward information. GRU achieves a comparable performance to LSTM with less computational overhead~\cite{cho2014properties}. Then, the produced final state vectors from both directions, $\bm{s}_{f}$ and $\bm{s}_{b}$, are concatenated and linearly transformed to be a latent vector $E(x)$. $m$ is the hidden state dimension for the GRU function.
\begin{equation}
\label{equ:enc}\
E(x) = \bm{W}_{h} \times \concat(\bm{s}_{f}, \bm{s}_{b}),  
\end{equation}
$ \textrm{where}\ \bm{s}_{f}, \bm{s}_{b} \in \mathbb{R} ^ {m},  \ \bm{W}_{h} \in \mathbb{R} ^ {h \times 2m}.$\\

The generator is another recurrent neural network with GRU. It generates a text token-by-token. For each timestamp $i$, it calculates a logit weight $z_{iv}$ for every candidate token $v \in \mathcal{V}$, conditioned on the latent vector, last original token $x_{i-1}$, and the last hidden state $\bm{s}_{i-1}$ of the GRU function.
\begin{equation*}
 \begin{split}
 z_{iv} & = \bm{w}^\top_v \gru(E(x), Em(x_{i-1}), \bm{s}_{i-1}) + b_v 
 \end{split}
\end{equation*}
Let $\xgi$ denote the random variable for the generated token at timestamp $i$. Its probability mass function is proportional to each candidate token's weight $z_{ti}$. This is modeled through a typical softmax function:
\begin{align}
\label{equ:sam1}
 Pr[\xgi = v] &=  \exp \left( z_{iv} \right) / \sum_{\vp \in \mathcal{V}} \exp \left( z_{i\vp} \right)  
\end{align}
For each timestamp $i$, a typical seq2seq model generates text by applying
$\argmax_{v \in \mathcal{V}}Pr[\xgi=v]$. However, this process does not protect the privacy of the original data.
\subsection{Differentially Privacy Text Sampling with Two-Set Exponential Mechanism}
To protect an individual's privacy and hide the authorship of the original input text, we couple differential privacy mechanism with the above sampling process in the generator. 
The \emph{exponential mechanism}~\cite{mcsherry2007mechanism} can be applied to both numeric and categorical data~\cite{fernandes-2018}. It is effective in various sampling process for discrete data. It guarantees privacy protection by injecting noise into the sampling process:

\theoremstyle{definition}
\begin{definition}{\textbf{Exponential Mechanism.}}
\label{def:expm}
Let $\mathcal{M}$ and $\mathcal{N}$ be two enumerable sets. Given a privacy budget $\epsilon$ $>$ 0, a rating function $\rho$: $\mathcal{M} \times \mathcal{N} \rightarrow \mathbb{R}$. The probability density function of the random variable $\varepsilon_{\epsilon, \rho}(m)$, $Pr\left[\varepsilon_{\epsilon, \rho}(m)=n\right]$ is:
\begin{equation}
\label{equ:expm}
\frac{\exp \left(\frac{\epsilon}{2 \Delta} \rho(m, n)\right)}  {\sum_{n^{\prime}} \exp \left(\frac{\epsilon}{2 \Delta} \rho\left(m, n^{\prime}\right)\right)}
\end{equation}
where $\Delta$, the sensitivity, means the maximum difference of rating function values between two adjacent datasets, and $m \in \mathcal{M}$, $n \in \mathcal{N}$.$\QEDA$
\end{definition}

However, according to Weggenmann and Kerschbaum~\shortcite{weggenmann2018syntf}, the exponential mechanism requires a large privacy budget to produce meaningful results while the output space is large, the vocabulary size in our case. It's nontrivial to randomly sample a good result directly among 20,000 candidates. 
To tackle the large output space issue, inspired by \textit{subsampled exponential mechanism}~\cite{lantz2015subsampled}, we propose a \textit{two-set exponential mechanism} to produce meaningful results with a better privacy protection. Instead of using a database independent distribution, we use a model-based distribution to generate subsets of tokens.

\theoremstyle{definition}
\begin{definition}{\textbf{Two-Set Exponential Mechanism.}}
\label{def:two_set}
Let $\mathcal{V}$ be a enumerable set with size $s$. Given the model-based probabilities of each item in $\mathcal{V}$, $Pr[v]$ for $v \in \mathcal{V}$, an item set $S$ of size $k$ is built by repeatedly sampling proportional to $Pr[v]$ with replacement. Other items are denoted as set $O$, where $\mathcal{V} = O \cup S$, $\emptyset = O \cap S$. 
Let $\mathcal{N} = \{S, O\}$. An item set $C_{dp}$ is chosen from $\mathcal{N}$ through the exponential mechanism with a rating function $\rho$: $\sum_{v \in C} Pr[v] / \sum_{C' \in \mathcal{N}, v' \in C'} Pr[v']$. Given $\epsilon > 0$, $N \in \mathcal{N}$, the probability density function (PDF) of the random variable $\varepsilon_{\epsilon, \rho}(C)$, $Pr\left[\varepsilon_{\epsilon, \rho}(C) = N \right]$, is:
\begin{equation}
\label{equ:twosetexpm2}
\frac{\exp \left(\frac{\epsilon}{2 \Delta} \rho(C, N)\right)} { \sum_{N^{\prime} \in \mathcal{N}} \exp \left(\frac{\epsilon}{2 \Delta} \rho\left(C, N^{\prime}\right)\right)}
\end{equation}
After choosing the set, $C_{dp}$, an item is randomly picked from the chosen set: $v \sim Random(C_{dp})$. Thus, given $v, w \in \mathcal{V}$,  $Pr[\varepsilon_{\epsilon, \rho}(v) = w]$ is:
\begin{equation*}
\label{equ:twosetexpm}
\begin{split}
 &Pr[w S] * Pr[\varepsilon_{\epsilon, \rho}(C)=S|w S] * Pr[w |w S, S]+
    \\ &Pr[w O] * Pr[\varepsilon_{\epsilon, \rho}(C)=O|w O] * Pr[w |w O, O],
\end{split}
\end{equation*}
where $Pr[wS]$ and $Pr[wO]$ are respectively the probability of $w$ in set $S$ and $O$, $Pr[wO] = (1-Pr[w])^{k}$. $\QEDA$
\end{definition}
\begin{theorem} \textbf{Two-Set Exponential Mechanism.\footnote{The proof is provided in Appendix~\ref{apdx:two_sum_proof}}}
\label{theom:two-set}
Given a privacy budget $\epsilon > 0$ and the size of output space $s$, two-set exponential mechanism is $(\epsilon + \ln{(s)})$-differentially private.$\QEDA$
\end{theorem}

By plugging our model with this mechanism,
we have the probability mass function for $\xgdpi$: $Pr[\xgdpi=tk]$.
This function models the disturbed probability distribution for all the alternative token $tk$ to replace the original variable. According to Theorem~\ref{theom:two-set}, sampling from $\xgdpi$ for each timestamp $i$ is $(\epsilon + \ln{(s)})$-differentially private.
Recall that in Definition~\ref{def:def}, the timestamp is bound by $l$. To generate text $\xgdp$, the generator samples a token for timestamp $i$ through the chosen set $T_{i}$:
\begin{equation}
\label{equ:samfull}
    \xgdp[i] \sim Random(T_{i}) \;\;\;\; \text{for}\;i\in [1,l]
\end{equation}

The \emph{composition theorem}~\cite{dwork2014algorithmic} is an extension to differential privacy. By repeating $n$ $\epsilon$-differentially private algorithms, the complete process achieves an $\epsilon n$-differential privacy. Algorithm~\ref{alg:er-ae} shows the differentially private text generation of ER-AE. As proved in Appendix~\ref{apdx:dp_proof}:
\begin{theorem} \textbf{Differentially Private Text Sampling.}
\label{the:dpts}
Given a privacy budget $\epsilon$ $>$ 0, a sequence length $l$ $>$ 0, the generator's sampling function in Eq.~\ref{equ:samfull} is $(\epsilon + \ln{(s)}) * l$-differentially private. $\QEDA$
\end{theorem}

\subsection{Initial Grammar and Semantic Preservation}
To generate a human-friendly text that has a close semantic to the original one, we need to have a high-quality rating function $\rho_i$ for Eq.~\ref{equ:twosetexpm2}. 
This is achieved by training the ER-AE model's encoder to extract semantic information, and its generator to learn the relationships among the tokens for prediction. We follow an unsupervised learning approach since we do not assume any label information. First, we adopt the reconstruction loss function:
\begin{equation}
\label{equ:recon}
\mathcal{L}_{recon}=\sum_{x_i \in x, x \in \mathbb{D}}-\log Pr\left[\xgi = x_i \right]
\end{equation}
It maximizes the probability of observing the original token $x_i$ itself for the random variable $\xgi$. In the recent controllable text generation models, the reconstruction loss plays an important role to preserve grammar structure and semantics of input data~\cite{shetty20184} when combined with the other loss.

\subsection{REINFORCE Training for Semantic Augmentation}
Diving into the optimization aspect of the softmax function, the reconstruction loss function above encourages the model to produce a higher probability on the original token while ignoring the rest candidates. It does not consider the other tokens that may have a similar meaning under a given context. This issue significantly limits the variety of usable alternative tokens. Additionally, this loss function relies on a single softmax function for multi-object learning, it cannot provide the expressiveness required by the language model~\cite{yang2017breaking}. 
We inspect the candidates and in most of the cases, only the top-ranked token fits the context in the text. This is problematic because the mechanism for our sampling process also relies on the other candidates to generate text.
To address the above issue, we propose a novel embedding reward function using the pre-trained word embeddings. 
Word representation learning models
show that discrete text tokens' semantic can be embedded into a continuous latent vector space. The distance between word embedding vectors can be a reference to measure the similarity between different words. To encourage our rating function $\rho_i$ to learn richer and better substitute tokens, we propose a reward function that leverages the semantics learned from the other corpus. The text dataset to be anonymized and released can be small, and the extra semantic knowledge learned from the other corpus can provide additional reference for our rating function.
This reward function is inspired by the Policy Gradient loss function~\cite{sutton2000policy}, $\mathcal{L}_{embed}$ is:
\begin{equation*}
\begin{split}
-\sum_{x_i \in x, x \in \mathbb{D}} & \Big( \sum_{v \in  \mathbb{E}_k(\xgi)} {\log(Pr[\xgi = v])\gamma(x_{i}, v) } \\
&+ \sum_{w \sim \mathbb{V}_{k}}{\log(Pr[\xgi = w])\gamma (x_{i}, w)}\Big)
\end{split}
\end{equation*}
Generally, this reward function assigns credits to the under-rated tokens in the reconstruction loss function. Recall that $\mathbb{D}$ is the original dataset and $x$ is one of its texts. At time step $i$, this reward function first assigns rewards to the top-$k$ selected tokens, denoted as $\mathbb{E}_k(\xgi)$, according to probability estimates for random variable $\xgi$ in Eq.~\ref{equ:sam1}.
The rewards are proportional to their semantic relationship to the original token $x_i$. It is defined as a function $\gamma: \mathcal{V} \times \mathcal{V} \rightarrow \mathbb{R}$, $\gamma (w, v)$ is:
\begin{equation}
\min\big(\cosine (Em(w), Em(v)), 0.85\big)
\end{equation}
The $min$ function avoids the generator focusing only on the original token. By assigning rewards to $\mathbb{E}_k(\xgi)$, it encourages the other candidates having a close semantic to the targeted one, but it may fail to reach infrequent tokens. Therefore, in the second part of the reward function, we encourage the model to explore less frequent tokens by random sampling candidates as $\mathbb{V}_{k}$.
This design balances the exploitation (top-$k$) and the exploration ($\mathbb{V}_{k}$) in reinforcement learning.

During training, the model is firstly pre-trained by minimizing the reconstruction loss  in Eq.~\ref{equ:recon} through the \emph{Adam} optimizer, and adopts the embedding reward loss later. The total loss is
\begin{equation}
\label{total_loss}
\mathcal{L} = \lambda_{recon} \times \mathcal{L}_{recon} + \lambda_{embed} \times \mathcal{L}_{embed}
\end{equation}
Specifically, the reconstruction loss can lead the model to generate grammatically correct text, and the embedding reward loss encourages the model to focus more on semantically similar tokens. The balance of the two loss functions are controlled by $\lambda_{recon}$ and $\lambda_{embed}$.

\begin{table*} [h]
\small
\caption{Results for each evaluation metric on both datasets. $\uparrow$ indicates the higher the better. $\downarrow$ indicates the lower the better.}
\label{result_table}
\centering
\begin{tabular}{l|ccc|cc}
\toprule
\multicolumn{1}{c}{ }  & \multicolumn{3}{c|}{Yelp (100-author)}    & \multicolumn{2}{c}{Conferences' Dataset}                              \\
\toprule
Model                & \multicolumn{1}{c}{USE $\uparrow$} & \multicolumn{1}{c}{Authorship $\downarrow$} & \multicolumn{1}{c|}{Stylometric$\uparrow$} & \multicolumn{1}{c}{USE$\uparrow$} & \multicolumn{1}{c}{Stylometric$\uparrow$}  \\
Original text    & 1     & 0.5513  & 0 & 1            & 0           \\

Random-R & 0.1183 & 0.0188  & 62.99  & 0.1356 & 65.624  \\
AE-DP & 0.6163 & 0.097 & 11.443 & 0.614 & 9.859 \\
SynTF \shortcite{weggenmann2018syntf}             & 0.1955 & 0.0518 & 26.3031  & 0.2161 & 25.95    \\
ER-AE (ours)    & 0.7548 & 0.0979 & 13.01  & 0.7424 & 9.838
\\ \toprule

\end{tabular}
\end{table*}

\begin{table}[!ht]
\begin{minipage}{.9\linewidth}
\centering
\small
\caption{The intermediate result of top five words and their probabilities at that the third and the forth generation steps.}
\label{ir}
\begin{tabular}{ll}
    \toprule
      \multicolumn{2}{c}{\textbf{Input:} there are \textbf{several unique} hot dog entrees to choose.}  \\
      \toprule
      & \multicolumn{1}{c}{\textbf{several}}                                              \\
AE-DP    & \textbf{several 0.98}, those 0.007, some 0.003                        \\
      & various 0.002, another 0.001                                      \\
ER-AE & many 0.55, some 0.20, \textbf{several 0.14}                               \\
      & different 0.04, numerous 0.03         \\  
\toprule 
      & \multicolumn{1}{c}{\textbf{unique}}                                              \\
AE-DP    & \textbf{unique 0.99}, different 0.0001, new 3.1e-05,                     \\
       & nice 2.5e-05, other 2.1e-05                                      \\
ER-AE & \textbf{unique 0.37}, great 0.21, amazing 0.15,       \\
     & wonderful 0.1, delicious 0.05        \\  
\toprule                    
\end{tabular}
\end{minipage}\hfill
\begin{minipage}{.9\linewidth}
\centering
\small
\caption{The estimated probability of a good candidate sampled with different mechanisms.}
\label{tl_test}
\begin{tabular}{lll}
    \toprule
       \multicolumn{3}{c}{\textbf{Input:} there are \textbf{several unique} hot dog entrees to choose.}  \\
       \toprule  
       & \multicolumn{1}{c}{\textbf{several}}            & \multicolumn{1}{c}{\textbf{unique}}    \\
Exponential Mechanism    & \multicolumn{1}{c}{0.0017} &   \multicolumn{1}{c}{0.00091}  \\ 
 Two-Set Exponential Mechanism      & \multicolumn{1}{c}{0.7411} &   \multicolumn{1}{c}{0.6794}   \\
 \toprule                    
\end{tabular}
\end{minipage}
\end{table}

\begin{table*}[h]
\small
\caption{Sample sentences generated by models.}
\label{sample}
\centering

\begin{tabular}{ll}

\toprule
\textbf{Input} & the play place is pretty fun for the little ones . \\
Random-R & routing longtime 1887 somalia pretty anatomical shallow the dedicated drawer rosalie \\
AE-DP & employer play lancaster mute fish fun for wallace little chandler . \\
SynTF & conditioned unique catherine marquis governing skinny garment hu vivid . insists \\
ER-AE & the play place is pretty nice with the little ones ! \\

\toprule
\textbf{Input} & i also ordered a tamarind margarita and it was great . \\
Random-R & substantial char recommended excavation tamarind coil longitudinal recover verify great housed \\
AE-DP & intersection also ordered service tamarind drooling scratched denis monkfish motions .  \\
SynTF & carnage spence unsigned also clinging said originated beacon liking strike accomplishments \\
ER-AE & i also requested a tamarind margarita and it were great . \\
\toprule
\textbf{Input} & i 'm not complaining because you do get exactly what you pay for . \\
Random-R & substantial char recommended excavation tamarind coil longitudinal recover verify great housed \\
AE-DP & comic-book 'm not mins because you donnelly get exactly tenderloin nerves bottomless for aldo box  \\
SynTF & leaf penetrated amounted jolted courageous socket fades unwilling tu judges regional numbering \\
ER-AE & i 'm not disappointing because you do make occult what you pay for . \\

\toprule
\textbf{Input} & the manuscript is well written is provides good insight into the problem . \\
AE-DP & the fig2c is well l102-103 wish provides horseshoe insight into the problem compositionality   \\
SynTF & ness voice incoming depending entrances somehow priscilla rows romantic oblivious mall \\
ER-AE & the manuscript is well edited has provides excellent insight into the problem .  \\

\toprule
\textbf{Input} & in particular , the generality of the approach is very well presented . \\
SynTF & wife pierced rotate specialist probe elects prussian beatty eccentric sweating . \\
ER-AE & in particular , this generality of an approach is very well written well \\

\toprule
\end{tabular}
\end{table*}

\begin{figure*}[h]
\centering
\includegraphics[width=0.8 \textwidth, height= 0.25\textwidth]{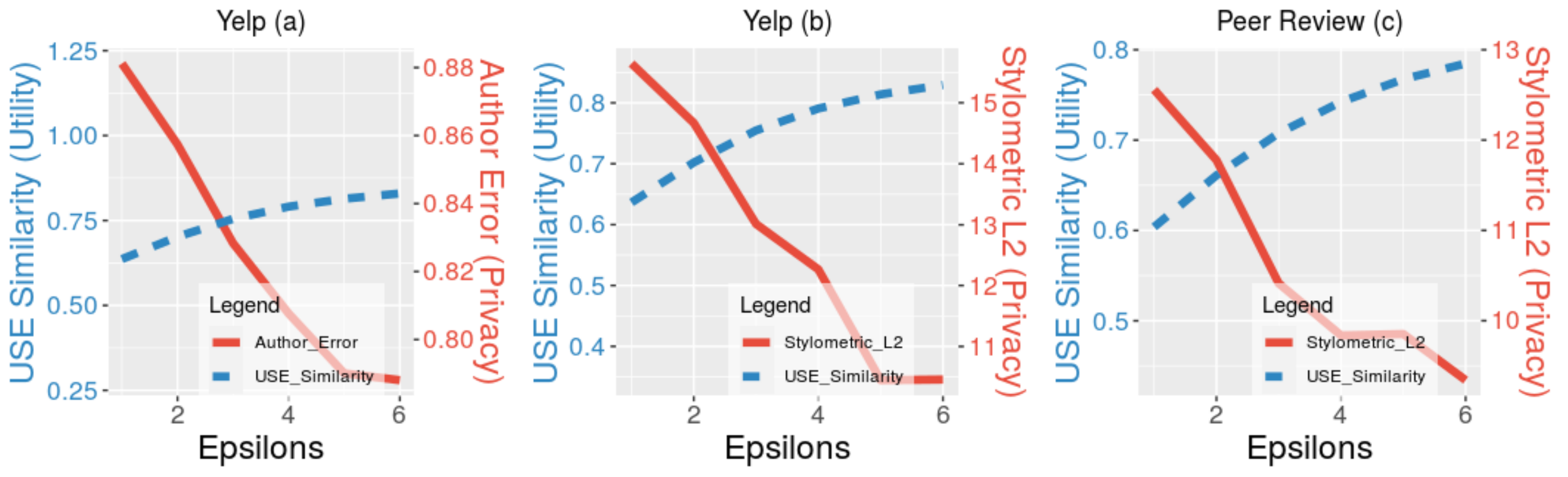}
\caption{Privacy v.s. Utility. Comparing USE similarity (utility), authorship identification error rate (privacy) and Stylometrics L2 distance (privacy) for different $\epsilon$s on applicable datasets.}
\label{fig:tradeoff}
\end{figure*}

\section{Experiment}
\label{others}

All the experiments are carried out on a Windows Server equipped with two Xeon E5-2697 CPUs (36 cores), 384 GB of RAM, and four NVIDIA TITAN XP GPU cards. We evaluate ER-AE on two different datasets with respect to its effectiveness for privacy protection and utility preservation.

\squishlist


\item
\textbf{Yelp Review Dataset}\footnote{http://www.yelp.com/dataset$\_$challenge}:
All the reviews and tips from the top 100 reviewers ranked by the number of published reviews and tips. It contains 76,241 reviews and 200,940 sentences from 100 authors.

\item
\textbf{Academic Review Dataset}:
All the public reviews from NeurIPS (2013-2018) and ICLR (2017) based on the original data and the web crawler provided by \cite{kang18naacl}. It has 17,719 reviews, 268,253 sentences, and the authorship of reviews is unknown. 

\squishlistend

Each dataset is divided into 70/10/20 for train/dev/evaluation respectively.
As mentioned in the related work discussion, most of the controllable text generation and style transferal studies rely on known authorship or other labels. Other generation models such as paraphrasing, however, hold an essentially different goal and cannot provide a privacy guarantee on the generated data. They are not applicable to our problem. Therefore, we pick SynTF~\cite{weggenmann2018syntf} and different generation and sampling models for evaluation:
\squishlist

\item \textbf{Random Replacement (Random-R)}: This method generates a new text by replacing each token in the text by randomly picking substitution from the vocabulary. 
\item \textbf{AE with Differential Privacy (AE-DP)}: Extended version of AE with the added two-set exponential mechanism for text generation. It does not include the embedding reward.
\item \textbf{SynTF~\cite{weggenmann2018syntf}}: We directly generate the tokens through SynTF's differentially private sampling function, without further extraction of the frequency vector.

\squishlistend
SynTF is a state-of-the-art generation model that satisfies differential privacy property on textual data. The other two simple baselines are for ablation test purposes. 

For ER-AE, we adopted a two-layers stacked GRU network for both the encoder and the generator. There are 512 cells in each GRU layer. The vocabulary size is 20,000, separately built for each dataset. All the word embeddings in our model come from the pre-trained \emph{BERT} embeddings provided by \cite{devlin2018bert}, which has a dimension of 768 for each embedding. The maximum input length of our model is 50, the learning rate is 0.001, the $k$ for embedding reward loss function is 5, the $\lambda_{recon}$ is 1, the $\lambda_{embed}$ is 0.5, and the batch size is 128. The $k$ in two-set exponential mechanism is 5. ER-AE is implemented in TensorFlow~\cite{abadi2016tensorflow}, and it uses the tokenizer in the NLTK library. Some traditional tricks for text generation, such as beam search, are not mentioned because they are incompatible with differential privacy. All the models are evaluated from three aspects: semantic preservation, privacy protection, and stylometric changes:
\squishlist
\item \textbf{Semantic Preservation (USE)}: A pre-trained \textit{Universal Sentence Embedding} (\textit{USE}) model\footnote{https://tfhub.dev/google/universal-sentence-encoder/1} from Google. It can embed a sentence into a latent vector that represents its semantics.
It is widely used for supervised NLP tasks such as sentiment analysis.
We measure the degree of semantic preservation using the cosine similarity between the latent vector of the original text and one of the generated text. 

\item \textbf{Privacy Protection (Authorship)}: A state-of-the-art authorship identification neural network model~\cite{sari2017continuous} to identify the authorship of generated text. The model is firstly trained on the training dataset, and the performance is evaluated on the testing set. The author's privacy is protected if s/he cannot be identified using authorship identification techniques.

\item \textbf{Stylometric Changes}: Well-established stylistic context-free features such as text length and a number of function words. We adopt \emph{StyloMatrix}~\cite{ding2017learning} for an aggregation of  features in 
\cite{iqbal2013unified,zheng2006framework}.
The feature vector change
is measured by the difference in L2 norm.
\squishlistend

\textbf{Quantitative Evaluation (Table~\ref{result_table}).} With a low utility (USE) score around 0.2 for both datasets, SynTF, and Random-R generate grammatically incorrect text and completely change the meaning of the original one. In contrast, ER-AE without semantic augmentation through REINFORCE training, denoted as AE-DP, achieves a much higher utility score of around 0.61. The full model ER-AE, with an $\epsilon$ of 3, achieves the highest utility score of 0.75 for Yelp reviews and 0.74 for peer reviews.
AE-DP, SynTF, and ER-AE all significantly reduce the chance of a successful authorship identification attack from 55\% to lower than 10\% in the Yelp data and introduce a variation in stylometric features of more than 10 in magnitude in the peer review dataset. They are all effective and competitive on removing the personal writing trait from the text data, but as mentioned above, AE-DP achieves the best and a much higher utility score. Although Random-R performs better on privacy protection, its generated texts are irrelevant to the original. 
Overall, with a competitive performance on anonymization, ER-AE performs significantly better than all of the other models on utility.

\textbf{Impact of Embedding Reward.} 
Table~\ref{sample} shows that the embedding reward plays an important role in selecting semantically similar candidates for substitution. 
AE-DP assigns a large probability to the original token and a tiny probability to the others. If applied with the mechanism, it is more likely to pick a semantically irrelevant token. 
ER-AE shows a smoother distribution and assigns higher probabilities to top-ranked semantically relevant tokens. Its generated candidates are better.

\textbf{Case Study.}
Table \ref{sample} shows that both SynTF and Random-R cannot generate human-friendly text.
Due to the issue of reconstruction loss function [\ref{equ:recon}],
AE-DP cannot substitute token with similarly semantic tokens and destroys the semantic meaning. ER-AE, powered by embedding reward, can substitute some tokens with semantically similar ones: ``written" is replaced by ``editted", and the whole sentence still makes sense. Besides, it can preserve the grammatical structure of the input. However, due to some missing information from word embeddings, the model would fail to generate good candidates for sampling.
The third sample replaces ``exactly" with `` occult".
ER-AE still performs way better than other models. 

\textbf{Utility vs. Privacy.}
The privacy budget $\epsilon$ controls the trade-off between privacy and utility. A larger $\epsilon$ implies better utility but less protection on privacy. However, this is a relative value that implies different degrees of risk given different problems~\cite{weggenmann2018syntf,fernandes-2018}.
As proved by Weggenmann and Kerschbaum~\shortcite{weggenmann2018syntf} a higher $\epsilon$ is intrinsically necessary for a large output space, in our case the vocabulary, to generate relevant text. In fact, we have already significantly reduced the optimal $\epsilon$ value of 42.5 used by Weggenmann and Kerschbaum~\shortcite{weggenmann2018syntf} to around 13, given the same dataset.
One possible way to lower the bound of $\epsilon$ is to directly factor in authorship and utility, such as topics, into the privacy model. However, it limits applicable to datasets.

\textbf{Exponential Mechanism vs. Two-Set Exponential Mechanism.}
In Table~\ref{tl_test}, we estimated the probability of a meaningful token (among top 5 semantically similar tokens) is sampled based on the intermediate probabilities in Table~\ref{ir}. Given a large output space of 20,000, the exponential mechanism is not likely to sample a meaningful token with a probability of 0.01 \%. However, the two-set exponential mechanism dramatically improves it from 0.01 \% to around 70\%. Our generator has a much higher chance to generate meaningful results with a similar privacy budget.

\section{Conclusion}
In this paper, we propose a novel model, ER-AE, to protect an individual's privacy for text data release. We are among the first to fuse the differential privacy mechanisms into the sequence generation process. We demonstrate the effectiveness of our model on the Yelp review dataset and two peer reviews datasets. 
However, we also find that ER-AE performs not very well on long texts due to the privacy budget accounting issue. Our future research will focus on improving long texts generation with better budget allocation scheme.

\section{Ethical Considerations}
Our model outperforms others on authorship obscuration and semantic preservation. Similar to other text generation tools, this model may be abused to generate fake reviews, but this can be assuaged by using fake review detection methods. This research work directly contributes to the area of privacy protection and indirectly promotes freedom of speech and freedom of expression in cyberspace.

\section*{Acknowledgments}
This research is supported in part by Natural Sciences and Engineering Research Council of Canada (NSERC) Discovery Grants (RGPIN-2018-03872), Canada Research Chairs Program (950-232791), and Provost Research Fellowship Award (R20093) and Research Incentive Funds (R18055) from Zayed University, United Arab Emirates. The Titan Xp used for this research was donated by the NVIDIA Corporation.







\bibliography{anthology,aaai}
\bibliographystyle{acl_natbib}

\appendix
\section{Proof of Differentially Private Text Sampling.}
\label{apdx:dp_proof}
\begin{theorem} \textbf{Differentially Private Text Sampling.}
Given a privacy budget $\epsilon$ $>$ 0, a sequence length $l$ $>$ 0, the generator's sampling function in Eq.7 is $(\epsilon + \ln{(s)}) * l$-differentially private. $\QEDA$
\end{theorem}
\begin{proof}
At the generation stage, for each timestamp $i$, our model generates a token by sampling from Eq.~7, which follows the form of exponential mechanism. This process achieves $(\epsilon + \ln{(s)})$-differential privacy as in Definition 4. Every input of the generator is the original input data $x_{i-1}$ (see Eq.2).
Eq.7 satisfies the sequential composition theorem. By repeating this process $l$ times, the complete sampling function provides $(\epsilon + \ln{(s)}) * l$-differential privacy. $\xgdp$ is $(\epsilon + \ln{(s)}) * l$-differentially private.  
\end{proof}

\section{Proof of Two-Set Exponential Mechanism. }
\label{apdx:two_sum_proof}

\begin{theorem} \textbf{Two-Set Exponential Mechanism.}
\label{theom:two-set}
Given a privacy budget $\epsilon > 0$ and the size of output space $s$, two-set exponential mechanism is $(\epsilon + \ln{(s)})$-differentially private. $\QEDA$
\end{theorem}

\begin{proof}
Given tokens sets $S$ and $O$, $\mathcal{V} = O \cup S$, $\emptyset = O \cap S$, and $\mathcal{N} = \{S, O\}$.
Let the choice, $C$, on $S$ and $O$ be $\epsilon$-differentially private, the sampling of an item in the chosen set be totally random. With $Pr[tk,tk N, N|x] = Pr[tk N| x] *Pr[\varepsilon_{\epsilon, \rho}(C)=N|tk N,x] * Pr[tk | N,tk N, x]$, where $N \in \mathcal{N}$, $i \in [1, l]$, $tk \in \mathcal{V}$, for any $x' \sim x$:
\begin{equation*}
\begin{split}
    &\frac{Pr[\varepsilon_{\epsilon, \rho}(\Tilde{x}_{i}) = tk | x]}{Pr[\varepsilon_{\epsilon, \rho}(\Tilde{x}_{i}) = tk | x']} =
    \\ & \qquad \qquad \frac{Pr[tk,tk S, S|x]}{Pr[tk,tk S, S|x'] + Pr[tk,tk O, O|x']} \\
    & \qquad + \frac{Pr[tk,tk O, O|x]}{Pr[tk,tk S, S|x'] + Pr[tk,tk O, O|x']}.
\end{split}
\end{equation*}

For the first part, denoted as $P_{S}$, with $\mathcal{V}$ of size $s$, by dividing the numerator and denominator with $Pr[\varepsilon_{\epsilon, \rho}(C)=S|tk S,x] * Pr[tk | S,tk S, x]$, we can get:
\begin{equation*}
    \begin{split}
        P_{S} &=\frac{Pr[tk,tk S, S|x]}{Pr[tk,tk S, S|x'] + Pr[tk,tk O, O|x']}  \\
        &\geq \frac{Pr[tk S| x] }{ \exp{(\epsilon)}*s}
    \end{split}
\end{equation*}
since 
\begin{equation*}
    \begin{split}
        \frac{Pr[\varepsilon_{\epsilon, \rho}(C)=S|tk S,x']}{Pr[\varepsilon_{\epsilon \rho}(C)=S|tk S,x]} 
        & \leq \exp{(\epsilon)} \\
        \frac{Pr[\varepsilon_{\epsilon, \rho}(C)=O|tk O,x']}{Pr[\varepsilon_{\epsilon \rho}(C)=S|tk S,x]}
        &\leq \exp{(\epsilon)} \\
        Pr[tk | S, tk S, x'] / Pr[tk | S, tk S, x] &\leq s \\
        Pr[tk | O, tk O, x'] / Pr[tk | S, tk S, x] &\leq s.
    \end{split}
\end{equation*}

For the second part, denoted as $P_{O}$, similarly, we have $P_{O} \geq {P[tk O| x] }/({ \exp{(\epsilon)}*s}).$ Then, 
\begin{equation*}
\begin{split}
    \frac{Pr[\varepsilon_{\epsilon, \rho}(\Tilde{x}_{i}) = tk | x]}{Pr[\varepsilon_{\epsilon, \rho}(\Tilde{x}_{i}) = tk | x']} &= P_{S} +  P_{O} \\ \geq& \frac{Pr[tk S| x] + Pr[tk O| x] }{ \exp{(\epsilon)}*s}
\end{split}
\end{equation*}

Since $Pr[tk S| x] + Pr[tk O| x] = 1$, the equation can be written as: ${Pr[\varepsilon_{\epsilon, \rho}(\Tilde{x}_{i}) = tk | x']}/{Pr[\varepsilon_{\epsilon, \rho}(\Tilde{x}_{i}) = tk | x]} \leq \exp{(\epsilon)}*s = \exp{(\epsilon + \ln{(s)})}$.
Therefore, the two-set exponential mechanism satisfies ($\epsilon + \ln{(s)}$)-differential privacy.
\end{proof}

\end{document}